\documentclass[a4paper,UKenglish,cleveref, autoref]{lipics-v2021}
\usepackage{graphicx} 
\usepackage{amsmath}
\usepackage{amssymb,amsthm}
\usepackage{xcolor}
\definecolor{red}{gray}{0}
\definecolor{blue}{gray}{0}

\usepackage[disable]{todonotes}

\usepackage{comment}
\usepackage{multirow} 
\usepackage{algorithm}  
\usepackage{algpseudocode}

\usepackage{caption}
\usepackage{subcaption}
\usepackage{lineno}

\newcommand{\OPT}{\textnormal{OPT}}
\newcommand{\ALG}{\textnormal{ALG}}

\newcommand{\partition}{S}

\usepackage{xspace} 
\newcommand{\lightness}{subset-lightness\xspace} 

\title{Subsetwise and Multi-Level Additive Spanners with Lightness Guarantees}

\author{Reyan~Ahmed}{University of Arizona, Tucson, United States}{abureyanahmed@arizona.edu}{}{}
\author{Debajyoti~Mondal}{University of Saskatchewan, Canada}{d.mondal@usask.ca}{}{}
\author{Rahnuma~Islam~Nishat}{Brock University, Ontario, Canada}{rnishat@brocku.ca}{}{}

\authorrunning{Reyan~Ahmed, Debajyoti~Mondal, Rahnuma~Islam~Nishat}

\Copyright{Reyan Ahmed, Debajyoti Mondal, Rahnuma Islam Nishat}



\begin{document}

\nolinenumbers
\maketitle



\begin{abstract} 
An \emph{additive +$\beta W$ spanner} of an edge weighted graph  $G=(V,E)$ is a subgraph $H$ of $G$ such that for every pair of vertices $u$ and $v$, $d_{H}(u,v) \le   d_G(u,v) + \beta W$, where $d_G(u,v)$ is the shortest path length from $u$ to $v$ in $G$. 
While additive spanners are very well studied in the literature, spanners that are both additive and lightweight have been introduced more recently [Ahmed et al., WG 2021]. 
Here the \emph{lightness} is the ratio of the spanner weight to the weight of a minimum spanning tree of $G$. In this paper, we examine the widely known subsetwise setting when the distance conditions need to hold only among the pairs of a given subset $S$. We generalize the concept of lightness to \lightness using a Steiner tree and provide polynomial-time algorithms to compute subsetwise additive $+\epsilon W$ spanner and $+(4+\epsilon) W$  spanner with $O_\epsilon(|S|)$ and $O_\epsilon(|V_H|^{1/3} |S|^{1/3})$ \lightness, respectively, where 
$\epsilon$ is an arbitrary positive constant.  
We next examine a multi-level version of spanners that often arises in network visualization and modeling the quality of service requirements in communication networks. The goal here is to compute a nested sequence of spanners with the minimum total edge weight. We provide an $e$-approximation algorithm to compute multi-level spanners assuming that an oracle is given to compute single-level spanners, improving a previously known 4-approximation [Ahmed et al.,  IWOCA 2023]. 
\end{abstract}
\maketitle 

\section{Introduction}

Given a graph $G$, a spanner of $G$ is a subgraph that preserves lengths of shortest paths in $G$ up to some multiplicative or additive error~\cite{peleg1989graph,liestman1993additive}. 
A subgraph $H=(V,E'\subseteq E)$ of $G$ is called a \emph{multiplicative $\alpha$--spanner}  if the lengths of shortest paths in $G$ are preserved in $H$ up to a multiplicative factor of $\alpha$, that is, $d_H(u,v) \le \alpha \cdot d_G(u,v)$ for all $u,v \in V$, where $d_G(u,v)$ denotes the length of the shortest path from $u$ to $v$ in $G$. For being an \emph{additive +$\beta$ spanner}, $H$ must satisfy the inequality $d_H(u,v) \le d_G(u,v)+\beta$. 

\begin{figure}[pt]
    \centering
    \includegraphics[width=\linewidth]{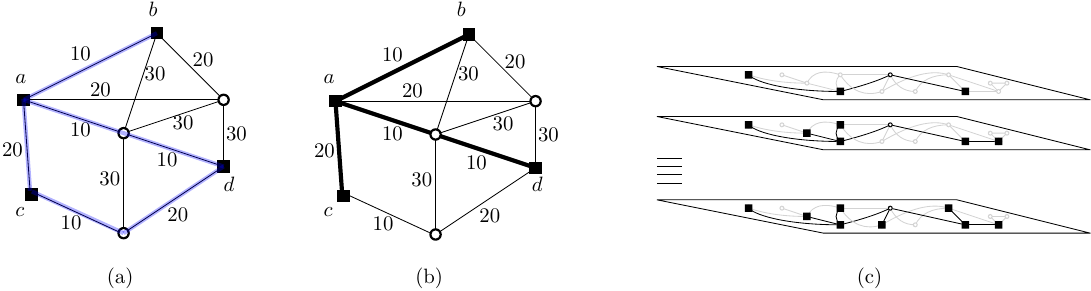}
    \caption{(a) A graph $G$, where the subset $S$ is shown in squares. A subgraph corresponding to the pairwise shortest paths determined by $S$ is highlighted in blue. Here $W_G(a,b)=W_G(a,d)=W_G(b,d)=10$ and $W_G(a,c)=W_G(b,c)=W_G(c,d)=20$. (b) A $+2W(\cdot,\cdot)$-spanner (in bold) of lightness 50/60$ = 0.83$ considering a spanning tree of the blue subgraph as the Steiner tree. (c) Illustration for a multi-level spanner.}
    \label{fig:intro}
\end{figure}

Multiplicative spanners were introduced by Peleg and Sch\"{a}ffer~\cite{peleg1989graph} in 1989. Since then a rich body of research has attempted to find sparse and lightweight  spanners on unweighted graphs. Such spanners find applications in computing graph summaries, building distributed computing models, and designing communication networks~\cite{awerbuch1985complexity,ahmed2020graph}. 
Finding a multiplicative $\alpha$--spanner with $m$ or fewer edges is NP--hard~\cite{peleg1989graph}. It is also NP--hard to find a $c \log |V|$-approximate solution where $c<1$, even for bipartite graphs~\cite{Kortsarz2001}.  
However, every $n$-vertex graph admits a  multiplicative $(2k-1)$-spanner with  $O(n^{
1+1/k})$ edges and  $O(n/k)$ lightness~\cite{Alth90}, which has been improved in subsequent research~\cite{Filtser16journal,ENS14,CW16journal,CDNS92,le2023unified}.

A multiplicative factor allows for larger distances in a spanner, especially for pairs that have larger graph distances. Therefore, an additive spanner is sometimes more desirable even with the cost of having a larger number of edges. For unweighted graphs, there exist additive $+2, +4$ and $+6$ spanners with $O(n^{3/2})$~\cite{Aingworth99fast,knudsen2014additive}, $\tilde{O}(n^{7/5})$~\cite{chechik2013new} and $O(n^{4/3})$ edges~\cite{baswana2010additive,knudsen2014additive}, respectively, whereas any additive $+n^{o(1)}$ spanners requires $\Omega(n^{4/3-\epsilon})$ edges~\cite{Abboud15}. Here $\tilde{O}()$ hides polylogarithmic factors. These results have been extended to weighted graphs with the additive factors $+\beta W_G(\cdot,\cdot)$, where $W_G(\cdot,\cdot)$ represents the maximum weight on a shortest path between $u$ and $v$ in $G$, i.e., $d_H(u,v) \le d_G(u,v)+\beta W_G(u,v)$. We use $W(u,v)$ instead of $W_G(u,v)$ if the graph $G$ is clear from context. 
Specifically, there exist additive $+2W(\cdot,\cdot), +4W(\cdot,\cdot)$ and $+(6+\epsilon)W(\cdot,\cdot)$ spanners of size  $O(n^{3/2})$~\cite{elkin2019almost}, $\tilde{O}(n^{7/5})$~\cite{ahmed2020graph} and $O_\epsilon(n^{4/3})$~\cite{elkin2023improved}, respectively, where $O_\epsilon()$ hides $poly(1/\epsilon)$ factors. Researchers have also attempted to guarantee a good lightness bound for such spanners, where the \emph{lightness} is the ratio of the spanner weight to the weight of a minimum spanning tree of $G$. So far, non-trivial lightness guarantees have been obtained only for (all-pairs) additive $+\epsilon W(\cdot,\cdot)$  and $+(4+\epsilon)W(\cdot,\cdot)$ spanners, where the lightness bounds are $O_\epsilon(n)$ and $O_\epsilon(n^{2/3})$, respectively~\cite{ahmed2021additive}.

Real-life graphs can be very large, which motivates computing a subgraph that approximately preserves the shortest path distances among a subset of important vertices. Subsetwise spanner (also known as terminal spanners~\cite{elkin2017terminal,bartal2019notions}) is another commonly studied version, 
where in addition to $G$,  the input consists of a vertex subset $S$ of $G$, known as terminals. The goal is to compute an additive or multiplicative spanner that guarantees that for every $u,v\in S$, $d_H(u,v)$ is within an additive or multiplicative factor of $d_G(u,v)$. Computing a subsetwise spanner comes with additional challenges, e.g., consider computing a minimum spanning tree vs. a minimum Steiner tree for an analogy. There exist additive $+2$ and additive  $+(2+\epsilon)W(\cdot,\cdot)$ subsetwise spanners with size $O(n\sqrt{|S|})$ and  $O_\epsilon(n\sqrt{|S|})$, respectively~\cite{ahmed2020graph}, however they do not come with any non-trivial lightness guarantees. The algorithms that construct such spanners often follow a \emph{greedy construction framework}, i.e., they construct an initial graph and then examine the pairs in sorted nondecreasing order of the maximum weight, i.e.,  $W(\cdot,\cdot)$. For each pair $(u,v)$, if it still fails to satisfy the distance condition of an additive spanner, then all the edges on the corresponding shortest path are added to the spanner. 

In this paper, we focus on computing subsetwise additive spanners that are also light with respect to a Steiner tree $T$. To this end, we define Steiner-lightness and \lightness.

{\color{blue}
\begin{definition}[Steiner-Lightness]
Let $G=(V,E)$ be a graph with terminal set $S$ and let $R$ be a minimum Steiner tree over $S$, i.e., $R$ spans all the vertices in $S$ but may include some vertices from $V\setminus S$. Then the \emph{Steiner-lightness} of a spanner in $G$ is the ratio of the spanner weight to the weight of $R$.
\end{definition}

Since a Steiner tree may not always satisfy additive shortest path distances between the pairs in $S$, we define another lightness measure as follows.
\begin{definition}[Subset-Lightness]
Let $G=(V,E)$ be a graph with terminal set $S$. Let $P \subseteq S\times S$ be the pairs that do not satisfy the spanner condition in a minimum Steiner tree with terminal set $S$. Let $T$ be a minimum Steiner tree with terminal set $(S\cup S^*)$, where $S^*$ is the set of vertices on the shortest paths of the pairs in $P$. Then the \emph{\lightness} is the ratio of the spanner weight to the weight of $T$.
\end{definition}

\smallskip
\noindent
\textbf{Our contribution.} {\color{red}A good  lightness measure should ideally be a good approximation for a minimum weight additive spanner. However, we show that neither Steiner-lightness nor  \lightness meets this property for all graphs.}
Specifically, there exist graphs where every 
$+cW(\cdot,\cdot)$ additive spanner has a Steiner-lightness of $\Omega(|S|)$ but a \lightness of $O(1)$ (Theorem~1). Similarly, there exist graphs where every minimum weight $+cW(\cdot,\cdot)$ additive spanner has a Steiner-lightness of $O(1)$ but a \lightness of $O(1/|S|)$ \todo{double chekc the minimum weight, $\Omega$ and $O$. $O(1/|S|)$ seems ok, but we might want to use $\Omega(1)$ instead of $O(1)$.}(Theorem~2).




}

We also construct additive spanners with good \lightness guarantees.  Let $\epsilon>0$ be a positive number, $G$ be a weighted graph and $S$ be a subset of vertices in $G$. We show the following (also see Table~\ref{tab:res}).
\begin{enumerate} 
\item $G$ has a subsetwise $+\epsilon W(\cdot, \cdot)$ spanner with $O_\epsilon(|S|)$ \lightness (Theorem~\ref{thm:eps_spnr}) and the spanner can be constructed deterministically.
\item $G$ has a subsetwise $+(4+\epsilon) W(\cdot, \cdot)$  spanner with $O_\epsilon(|V_H|^{1/3} |S|^{1/3})$ \lightness (Theorem~\ref{thm:4spnr}) and the spanner can be constructed deterministically.
\item $G$ has a subsetwise $+(4+\epsilon) W_{\max}$  spanner with $\tilde{O}_\epsilon(|S|\sqrt{|V'_H|/|V_H|})$ \lightness, where $W_{\max}$ is the maximum edge weight of the graph; $V_H$ and $V'_H$ are the sets of vertices of the Steiner trees connecting $S$ and a randomly sampled subset of $S$, respectively (Theorem~\ref{thm:4spnr_sampled}). 
\end{enumerate}

\begin{table}[h]
\caption{Results on lightweight additive spanners.}
\label{tab:res}
\footnotesize
\centering
\begin{tabular}{|ccp{0.25cm}p{.25cm}cp{0.25cm}|lcllp{0.25cm}l|}
\hline
\multicolumn{6}{|c|}{All-pairs Additive Spanners}                                  & \multicolumn{6}{c|}{Subsetwise Sp.}                                    \\ \hline
\multicolumn{3}{|c|}{Unweighted}                                                   & \multicolumn{3}{c|}{Weighted}                                                        & \multicolumn{6}{c|}{Weighted}                \\ \hline
\multicolumn{1}{|c|}{+$\beta$}            & \multicolumn{1}{c|}{Lightness}                & \multicolumn{1}{p{.01cm}|}{Ref}                                         & \multicolumn{1}{p{2cm}|}{+$\beta$}                                     & \multicolumn{1}{c|}{Lightness}                         & \multicolumn{1}{p{.01cm}|}{Ref}                   
 & \multicolumn{5}{p{.2cm}|}{Subset-Light.}                  &  Ref                                 \\ \hline
\multicolumn{1}{|c|}{0}                  & \multicolumn{1}{c|}{$O(n)$}             & \multicolumn{1}{p{.5cm}|}{\cite{khuller1995balancing}}                   & \multicolumn{1}{p{.8cm}|}{{+$\epsilon W(\cdot,\cdot)$}}                                            & \multicolumn{1}{c|}{{$O_\epsilon(n)$}}                                 & \multicolumn{1}{p{.5cm}|}{\cite{ahmed2021additive}}                                           & \multicolumn{5}{p{1.2cm}|}{$O_\epsilon(|S|)$}                          & Th~\ref{thm:eps_spnr}                                                     \\ \hline

\multicolumn{1}{|c|}{+2}                  & \multicolumn{1}{c|}{$O(n^{1/2})$}             & \multicolumn{1}{p{1.3cm}|}{\cite{Aingworth99fast,knudsen2014additive}}                   & \multicolumn{1}{p{.5cm}|}{-}                                            & \multicolumn{1}{p{.5cm}|}{-}                                 & \multicolumn{1}{p{.5cm}|}{-}                                           & \multicolumn{5}{p{.5cm}|}{-}                          & \multicolumn{1}{p{.5cm}|}{-}                                                     \\ \hline
 
\multicolumn{1}{|c|}{+4}                  & \multicolumn{1}{c|}{$O(n^{2/5})$}             & \multicolumn{1}{p{1cm}|}{\cite{chechik2013new}}                   & \multicolumn{1}{p{.8cm}|}{+(4+$\epsilon)$ $W(\cdot,\cdot)$}               & \multicolumn{1}{l|}{$O_\epsilon(n^{2/3})$}             & \multicolumn{1}{p{1cm}|}{\cite{ahmed2021additive}}                                           & \multicolumn{5}{p{1.6cm}|}{$O_\epsilon(|V_H|^{1/3}$ $|S|^{1/3})$}                          & Th~\ref{thm:4spnr},~\ref{thm:4spnr_sampled}                                                      \\ \hline
 
\multicolumn{1}{|c|}{+6}                  & \multicolumn{1}{c|}{$O(n^{1/3})$}             & \multicolumn{1}{p{1.3cm}|}{\cite{baswana2010additive,knudsen2014additive}}                   & \multicolumn{1}{p{.5cm}|}{-}                                            & \multicolumn{1}{p{.5cm}|}{-}                                 & \multicolumn{1}{p{.5cm}|}{-}                                           & \multicolumn{5}{p{.5cm}|}{-}                          & -                                                     \\ \hline
\end{tabular}
\end{table}

Notice that the results of Theorem~\ref{thm:eps_spnr} and Theorem~\ref{thm:4spnr} exactly match the bounds provided for all-pairs spanners~\cite{ahmed2021additive} if we set $S=V$. Hence, our results are not only answering stronger questions, but they are also directly comparable with the existing results. The main challenge to generalizing the all-pairs results to subsetwise results is that the former compares against the minimum spanning tree. 
A counting argument that is often used in additive spanner constructions is based on the number of \textit{improvements}, i.e., the number of improvements that may occur for all pairs of vertices is no more than an upper bound. Here, an improvement means that the difference in the shortest path between the pair before and after adding a set of edges is non-zero. 
However, this idea does not work in the subsetwise setting if we only compute a Steiner tree that spans $S$. A key contribution of our work is that we address this problem by computing Steiner trees that not only span $S$ but also the vertices in the shortest paths of some vertex pairs from $S$.

We next consider computing multi-level (multiplicative/additive) spanner, where the goal is to compute a hierarchy of (multiplicative/additive) spanners 
that minimizes the total number of edges or the total edge weight of all the spanners. Specifically, the input consists of a graph and some subsets $S_1, S_2, \ldots, S_k$ of its vertices where $S_1\subseteq S_2\subseteq \ldots \subseteq S_k$. The goal is to find a hierarchy of spanners such that the spanner at the $i$th level connects $S_i$ and for $i>2$, the spanner  $S_i$ includes the edges of the spanners $S_1,\ldots, S_{i-1}$. Multi-level spanners can model real-world scenarios where different levels of importance are assigned to different sets of vertices. Examples of such scenarios include modeling the quality of service requirements for the nodes in communication networks~\cite{Charikar2004ToN}, or visualizing a graph on a map when the details are revealed as the users zoom in~\cite{gray2023scalable}.  There exists a 4-approximation algorithm for this problem~\cite{ahmed2023multi}, and we give an improved $e$-approximation (Theorem~\ref{thm:multi-level}).

The rest of the paper is organized as follows. Section~\ref{sec:comp} discusses the properties of Steiner-lightness and \lightness. Section~\ref{sec:lightweight} discusses the results on lightweight additive spanners. 
Section~\ref{sec:multi-level} presents the $e$-approximation algorithm for multilevel spanners. Finally, Section~\ref{sec:conclusion} concludes the paper with directions for future research.


\section{Comparing Steiner-Lightness and Subset-Lightness}\label{sec:comp}
\setcounter{theorem}{0}

In this section we compare Steiner-lightness and \lightness. We first show that there exist graphs with terminal set $S$, where constant additive spanners have Steiner-lightness $\Omega(|S|)$ but \lightness $O(1)$. 

\begin{theorem}
For every fixed $c\ge 1$, there exists a graph $G$ with terminal set $S$ such that every $+cW(\cdot,\cdot)$ additive spanner of $G$ has a Steiner-lightness of $\Omega(|S|)$ 
but a \lightness of $O(1)$.
\end{theorem}

\begin{proof}
Consider a $2n$-vertex complete graph $G$ with vertex set $S$ and let $A,B$ be a partition of $S$ into equal-size subsets. Assume that the weight of each edge that is crossing the partition is $c$ and the weight of each remaining edge is $(c+\delta)$, where $\delta>0$.  A minimum spanning tree $M$ of $G$ is a star graph with weight {\color{red}$O(c|S|)$}.
~Let $G'$ be a graph obtained by subdividing each edge of $G$ with $(2c-1)$ division vertices and distributing the weight of the edge equally among the $2c$ new edges. It is straightforward to verify that $M$ determines a minimum Steiner tree $R$ in $G'$ with terminal set $S$. Consider a pair $a,b$, where $a\in A$, $b\in B$ and the edge $(a,b)$ does not belong to $R$. Then the weight of the shortest path $a,\ldots,b$ in $R$ is $2c+\delta > cW(\cdot,\cdot) = c\cdot \frac{c+\delta}{2c} =\frac{c+\delta}{2}$ \todo{$d_R(a,b) = 2c + \delta > 2c = c + c = c + cW(\cdot, \cdot) = d_G(a,b) + cW(\cdot, \cdot)$, is it better?} and thus does not satisfy the distance condition for a $+cW(\cdot,\cdot)$  spanner. There are $\Omega(|S|^2)$ pairs with one element in $A$ and the other in $B$, and every $+cW(\cdot,\cdot)$  spanner must take all these shortest paths. Therefore, the Steiner-lightness of such a spanner is $\Omega(|S|)$. 

Consider now a minimum Steiner tree $T$ with terminals $(S\cup S^*)$, where $S^*$ are the vertices on the shortest paths of the pairs for which the path on $R$ do not satisfy $+cW(\cdot,\cdot)$  additive spanner constraint. Then $|S\cup S^*| = {\color{red}\Theta(c|S|^2)}$
~and thus the weight of $T$ is also {\color{red}$\Theta(c|S|^2)$} and . Since the weight of a  $+cW(\cdot,\cdot)$  spanner is {\color{red}$O(c|S|^2)$}, its \lightness is constant. 
\end{proof}
{\color{red}We now provide another result that shows the contrast between Steiner-lightness and \lightness.}

\begin{theorem}
\label{thm:2}For every fixed $c\ge 1$, there exists a graph $G$ with terminal set $S$ such that every minimum weight  $+cW(\cdot,\cdot)$ additive spanner 
has Steiner-lightness $O(1)$ but subset-lightness $O(1/|S| )$.
\end{theorem}
 \begin{proof}
 
{\color{red}


We now show a scenario where Steiner-lightness is $O(1)$, however, subset-lightness is $O(1/|S| )$.
~Consider a graph having a vertex set $ V =(V_a \cup V_b) , V_a \cap V_b = \emptyset$, where the vertices of $V_a$ (similarly, $V_b$) form a complete graph.  We subdivide each edge with four division vertices and put a weight of 2 unit to each edge. We add two vertices $c^1_a$ and $c^1_b$, and them make them adjacent to all vertices in $V_a$ and $V_b$, respectively, with edges of weight $6+\epsilon$.
~We add two vertices $c^2_a$ and $c^2_b$, and them make them adjacent to all vertices in $V_a$ and $V_b$, respectively, with edges of weight 5. Figure~\ref{fig:open_prob_1}(a) illustrates a cycle with three vertices of $V_a$ (corner points). Figure~\ref{fig:open_prob_1}(b) illustrates the vertex $c^1_a$ and its connection to the three vertices of $V_a$.


Suppose there are two disjoint stars $S^1_a$  and $S^1_b$ (with center vertices  $c^1_a$ and $c^1_b$) that span $V_a$ and $V_b$ respectively, and there is an edge  $P^1$ of weight 1 connecting $S^1_a$ and $S^1_b$. Now $S^1_a \cup S^1_b \cup P^1$ forms a Steiner tree $T$.  Figure~\ref{fig:open_prob_1}(c) illustrates $T$ where $|V_a|=|V_b|=3$. Let $T$ be a minimum Steiner tree spanning the terminal set $S = (V_a \cup V_b)$.

Now, suppose $S^1_a$ and $S^1_b$ are not light enough to satisfy the additive errors of the vertex pairs of $V_a$ and $V_b$, respectively. There is another pair of slightly lighter stars, $S^2_a$ and $S^2_b$, that satisfies the additive error. Here  $S^2_a$ and $S^2_b$ are connected through a slightly heavier edge $P^2$ of weight 8. Now $S^2_a \cup S^2_b \cup P^2$ forms a non-optimal Steiner tree $T'$ (see Figure~\ref{fig:open_prob_1}(d)), but it is a $+cW$ additive spanner with $c=1$. 

To formally argue that $T'$ is a $+c W(\cdot, \cdot)$-spanner, we have to show that for all pairs of subset vertices ($V_a \cup V_b$) the shortest path in $T'$ satisfies the spanner condition. 
First observe that the shortest path distance between a pair of vertices  of $V_a$ (see Figure~\ref{fig:open_prob_1}(a)) is 10. Notice that $T'$ is a $+c W(\cdot, \cdot)$-spanner where 
$c = 1$. 
 If the pair of vertices is from $V_a$ (or $V_b$), then the shortest path in $T'$ is equal to the shortest path of the input graph. If one vertex of the pair is from $V_a$ and the other is from $V_b$, then the shortest path of that pair in the input graph is $13+2\epsilon$
 ~(see Figure~\ref{fig:open_prob_1}(c)). In order to satisfy the spanner condition, the shortest path in $T'$ can be at most $13+2\epsilon + c\cdot (6+\epsilon)$, i.e., $19 + 3\epsilon$
 ~when $c = 1$. 
~Since the shortest path distance in $T'$ is 18, in general, $T'$ is a $+c W(\cdot, \cdot)$-spanner. 
Since the weight of $T'$ is within a constant factor of the weight of  $T$,   the 
Steiner-lightness is $O(1)$. 

To show that the \lightness of any minimum weight additive spanner is  $O(1/|S|)$, we now consider a minimum Steiner tree $T^+$ that includes all the vertices of the shortest paths as well as those in $S$. 
Since every division vertex must have at least one edge adjacent to in $T^+$, the weight of $T^+$ is at least $\Omega(|S|^2)$, \todo{should we say $\Theta (|S|^2)$? I think $\Omega(|S|^2)$ is better.}as illustrated in Figure~\ref{fig:open_prob_1}(e).  
Hence the \lightness of any minimum weight additive spanner is at most  $weight(T')/weight(T^+) = O(1/|S|)$.\todo{j:I added minimum weight because otherwise the statement is not true. I updated the theorem and also the intro accordingly, but did not change theorem 1.}

We can generalize the example of Figure~\ref{fig:open_prob_1}. Let us assume we want to construct a \todo{should we add minimum weight? Yes, that will be consistent.} $+c W(\cdot,\cdot)$-spanner for $c>1$. 
Instead of having the shortest path distance between any pair of vertices in $V_a$ (and $V_b$) equal to 10, we can consider that the shortest path distance is $x>0$. Similarly, instead of having the edge weight of an edge of $S_a^1$ equal to 6, we can set it $\frac{x}{2} + \frac{c'}{2}W_a$ 
where $W_a$ is the maximum edge weight of the shortest path between a pair from $V_a$ (or $V_b$) and $c'$ is slightly larger than $c$ (for example, $c' = c + 0.01$). Also, we set the weight of $P^1$ equal to $\eta$ where $\eta$ is slightly larger than zero (for example, $\eta=0.01$). 
Finally, we need to set the weight of $P^2$ equal to $y$ such that:
\begin{itemize}
    \item The graph $S_a^2 \cup S_b^2 \cup P^2$ is a $+c W(\cdot, \cdot)$-spanner.
    \item the weight of $T'$ is larger than the weight of $T$.
\end{itemize}
~Now, we get the condition $y \leq c' W_{ab} + \eta$ 
~to claim $T'$ is a $+c W(\cdot,\cdot)$-spanner where $W_{ab}$ is the maximum edge weight of the shortest path between a pair of vertices from $V_a$ and $V_b$. And we get the condition $|V_a| c' W_a + \eta < y$ 
~to claim that the weight of $T'$ is larger than the weight of $T$. In other words, $|V_a|W_a < \frac{y-\eta}{c'} \leq W_{ab}$. 
~We can set $W_a$ arbitrarily small and $W_{ab}$ arbitrarily large to satisfy this condition for any positive $y>\eta$ and $c'$.

\begin{figure}[pt]  
\centering
\includegraphics[width=\textwidth]{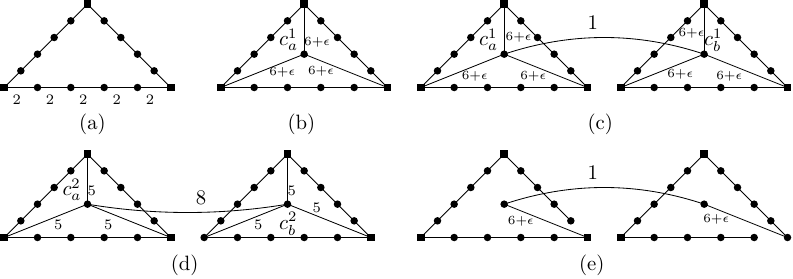}



\caption{Illustration for Theorem~\ref{thm:2}. (a) A cycle through the three vertices of $V_a$ which are located at the corners. (b) The star $S^1_a$. (c) The minimum Steiner tree $T$ ($S^1_a \cup S^1_b \cup P^1$). (d) A non-optimal Steiner tree $T'$ ($S^2_a \cup S^2_b \cup P^2$). (e) The minimum Steiner tree $T^+$ over $S$ and the shortest path vertices.
}\label{fig:open_prob_1}
\end{figure}
}
 \end{proof}


\section{Lightweight Additive Subsetwise Spanner}\label{sec:lightweight}
In this section, we give our results on additive and lightweight spanners. All the spanners that we construct rely on the following preliminary setup.

\subsection{Preliminary Setup}
 
Let $G=(V,E)$ be a weighted graph and let $S \subseteq V$ be a set of given vertices. 
Fix a shortest path for each pair of vertices in $S$. First, we compute a constant-factor approximation of the Steiner tree $R$ with $S$ as the set of terminals. Next, we identify the vertex pairs from $S$ for which the computed Steiner tree does not provide a path that satisfies the spanner condition. We then compute $S'$ by taking the union of $S$ and the vertices from the shortest paths of those vertex pairs. To compute a lightweight subsetwise spanner of $S$, we compute $R$ and $T$, which are constant-factor approximations for the Steiner trees with terminal sets $S$ and $S'$, respectively. Note that an approximate Steiner tree can be computed in polynomial time~\cite{Byrka2013}. 
We then construct a Steiner tree $H=(V_H, E_H)$ in $G$   with the set $S'$ as terminals by taking the union of $R$ and $T$ and discarding unnecessary edges from $T$. The reason for keeping the edges of $R$ is that later in our algorithm will use the vertices of $S'$, which we determined using $R$. Since $S\subseteq S'$, we have $weight(R)<weight(T)$, and hence  $weight(H)\le 2\times weight(T)$. Therefore, we can compare asymptotic \lightness with respect to $weight(H)$ instead of $weight(T)$.  
We will focus our attention on a weight-scaled version of $G$ and a subdivision of $H$, which is described below. 

After computing $H$, we multiply all edge weights of $G$ by $|V_H|/weight(H)$ to obtain a new weighted graph $G_s=(V_s, E_s)$. Let $H_s$ be the scaled Steiner tree obtained from $H$. We now observe some properties of $G_s$  that will be useful for the spanner construction. Although the Steiner tree has not been used previously for constructing lightweight spanners, such scaled graph has previously been used in~\cite{ahmed2021additive} for constructing lightweight spanners.  The following lemma shows this weight scaling preserves some distance relations.

\begin{lemma}\label{lem:weight_scaling} 
    Let $G'=(V', E')$ be a subgraph of $G$ and let $G'_s=(V'_s, E'_s)$ be a subgraph of $G_s$ such that $V'_s$ and $V'$ (similarly,  $E'_s$ and $E'$) correspond to the same set of vertices (edges) in $G$.  Then $d_{G'}(u,v) \leq \alpha d_G(u,v) + \beta W_G(u,v)$ if and only if $d_{G'_s}(u,v) \leq \alpha d_{G_s}(u,v) + \beta W_{G_s}(u,v)$, where $\alpha,\beta\ge 1$.
\end{lemma}
\begin{proof}
One can multiply $|V_H|/weight(H)$ to both sides of the inequality $d_{G'}(u,v) \leq \alpha d_G(u,v) + \beta W_G(u,v)$ to obtain $d_{G'_s}(u,v) \leq \alpha d_{G_s}(u,v) + \beta W_{G_s}(u,v)$. 
\end{proof}

The following lemma shows that we can safely remove the edges of weight larger than $|V_H|$ from $G_s$. 
\begin{lemma}\label{lem:remove_large_weight}
    Removal of the edges with weight larger than $|V_H|$ from $G_s$ does not increase the shortest path length for any pair of vertices of $S$.
\end{lemma}
\begin{proof}
    Assume for a contradiction that $e$ is an edge with weight larger than $|V_H|$ and removal of $e$ increases the shortest path length of $u,v \in S$. It now suffices to show that $e$ cannot be an edge on the shortest path between $u$ and $v$. Note that the Steiner tree $H_s$ spans $u$ and $v$. Therefore, the length of the shortest path between $u$ and $v$ in $G_s$ is at most the total edge weight of the scaled Steiner tree, i.e., $weight(H)$ multiplied by $|V_H|/weight(H)$, which is equal to  $|V_H|$ and hence the shortest path cannot contain $e$. 
\end{proof}

We now construct a tree $H'=(V_{H'},E_{H'})$ by subdividing the scaled Steiner tree $H_s$, which will help extend techniques for spanner construction to guarantee \lightness (Figure~\ref{fig:initialization}). 
Specifically, we subdivide each edge of $H_s$ using a minimum number of subdivision vertices such that all edge weights in $H'$ become less than or equal to one unit. 


\begin{figure}[pt]
 \centering\includegraphics[width=.8\linewidth]{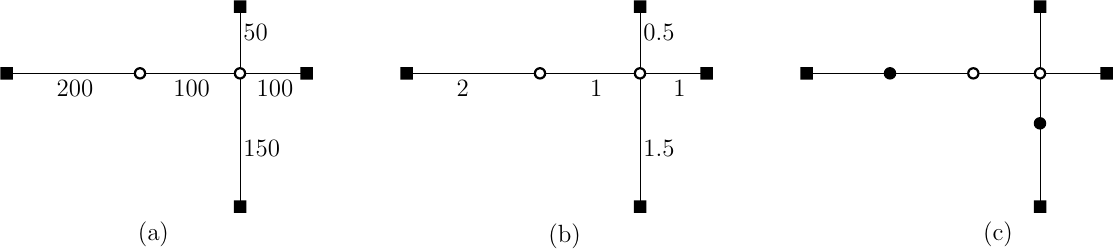}
     
\caption{Illustration for (a) Steiner tree $H$, (b) scaled Steiner tree $H_s$, and (c) subdivided Steiner tree $H'$, where the  
Squares represent vertices of $S$ and solid circles represent subdivision vertices.} \label{fig:initialization}
\end{figure}
\begin{lemma}
    Let $n_{H'}$ be the number of subdivision vertices in $H'$. Then $n_{H'} \leq weight(H')$. and $|V_{H'}| \leq 2|V_H|$. 
\end{lemma}
\begin{proof}
It suffices to subdivide each edge $(u,v)$ with $\lceil weight_{H_s}(u,v)\rceil -1$ vertices. 
The sum over all edges is bounded by  $weight(H_s){=}weight(H)$. 
Since $weight({H_s})\le |V_H|$, we have $|V_{H'}|=|V_H|+n_{H'}\le |V_H|+weight(H_s)\le 2|V_{H'}|$.  
\end{proof}


\subsection{A $+\epsilon W(\cdot, \cdot)$-spanner}
 

\textbf{Overview of the Construction.} We construct the spanner incrementally by maintaining a current graph. We first take a set of Steiner tree edges (with $S'$ as terminals) into the current graph and then examine each pair $(u,v)$ from $S$ in increasing order of $W(u,v)$. If the shortest path distance between $u,v$ in the current graph is more than what it requires for a $+\epsilon W(\cdot,\cdot)$-spanner, then we add all the edges on the shortest path between $u$ and $v$ in $G_s$ that are currently missing. During this step, we show that some pairs of vertices get  `set-off', i.e., their shortest path distances in the current graph become comparable to that of $G_s$, and the number of such pairs is comparable to the total weight of the edges that have been added in this step. The final bound on the \lightness of the spanner comes by relating the total weight added to the total number of vertex pairs that were set off during the whole construction. A similar technique has previously been used to compute lightweight spanner for the whole graph, and using minimum spanning tree~\cite{ahmed2021additive}. We extend that technique for subsetwise spanner using the Steiner tree.     


\smallskip
\noindent
\textbf{Construction Details.} For each vertex $v\in S'$, let $e_v$ be the minimum weight edge incident to $v$ in $G_s$. Let $H_0=(V_{H_0}, E_{H_0})$ be a graph where $S' \subseteq V_{H_0}$ and $E_{H_0}$ is obtained by choosing from all such edges the ones with a weight of less than one unit. 

\begin{lemma}\label{lem:neiborhood_small_weight}
    Let $P$ be a shortest path in $G_s$ between a pair of vertices $u, v \in S$. Let $e=(w,r)$ be an edge of $P$ that does not belong to $H_0$ such that $weight(e)<1$. Then there exists a vertex $q$ in $H_0$ such that $d_{H_0}(w,q) \leq weight(e)$. 
\end{lemma}
\begin{proof}
    Since $w \in S'$, by construction, $H_0$ contains an edge $e'=(w,q)$ incident to $w$ such that $weight(e')<weight(e)<1$. 
\end{proof}

\begin{lemma}\label{lem:first_ngbr}
    The number of vertices of $H_0$, $|V_{H_0}|\leq 2|S'|$ and 
    $weight(H_0) \leq |S'| \leq |V_H|$.
\end{lemma}
\begin{proof}
   Since $H$ is a Steiner tree on terminal set $S'$, $|S'| \leq |V_H|$. Since for each vertex $u \in S'$, we add at most one edge with weight less than one in $H_0$, $|V_{H_0}|\leq 2|S'|$ and $weight(H_0) \leq |S'|$.
\end{proof}


Let $H_1$ be the graph obtained by adding all edges of $H_0$ and $H'$. The following claim shows that any edge in the shortest path in $G_s$ if does not exist in $H_1$, then the end vertices of that edge has a close neighborhood (the set of neighbor vertices that has distance no more than the weight of the missing edge) in $H_1$. 

\begin{lemma}\label{lem:neighborhood_all_weight}
    Let $P$ be a shortest path in $G_s$ between a pair of vertices $u, v \in S$. Let $e=(w,r)$ be an edge of $P$ that does not belong to $H_1$. Then there exists a set of vertices $V_r \subseteq V_{H_1}$ such that $|V_r| = \Omega(weight(e))$ and for each vertex $q \in V_r, d_{H_1}(w,q) \leq weight(e)$.  
\end{lemma}
\begin{proof}
    Note that $w\in S'$ and hence a vertex of $H'$. If $weight(e)<1$, then the claim follows form Lemma~\ref{lem:neiborhood_small_weight} if we choose $V_r=\{q\}$ where $q$ is the neighbor of $w$ added in $H_0$. 
    Assume now that $weight(e) \geq 1$. By Lemma~\ref{lem:remove_large_weight}, $weight(e) \leq |V_H| \leq weight(H_s)$ 
    and $H'$ is obtained by subdividing the edges of $H_s$ such that the weight of each edge of $H'$ becomes at most one. Since $H_1$ contains all edges of $H'$, we can set $V_r$ to be the set of vertices that has a distance less than or equal to $weight(e)$ in $H_1$. Since the weight of edges of $H'$ is at most 1, $|V_r| = \Omega(weight(e))$.
\end{proof}

We are now ready to describe the algorithm. 
We create a graph $G'_s=(V'_s, E'_s)$ by removing the edges $E_{H_s}$ from $G_s$ and adding the edges of $H'$. The algorithm sorts the pairs of vertices in $S$ based on their increasing order of maximum edge weight in the shortest path on $G'_s$. 
Then it checks whether the distance condition is unsatisfied for pairs of vertices in this sorted order, and if so then adds the missing edges on the corresponding shortest path. Specifically, let $H_i$ (initially, $H_1$, which consists of all the edges of $H_0$ and $H'$) be the current subgraph of $G'_s$. Then if the distance condition is unsatisfied, then the algorithm constructs $H_{i+1}$ by adding the missing edges of the shortest path to $H_i$. Once all the unsatisfied pairs are processed, the spanner condition is satisfied for all pairs of vertices in $S$. Although we used $G'_s$ to check the spanner condition, where some edges may be subdivided, the subdivision does not change the shortest path length, and hence the final spanner corresponds to a valid $+\epsilon W(\cdot,\cdot)$ spanner. 

\begin{figure}[h]
\includegraphics[width=\linewidth]{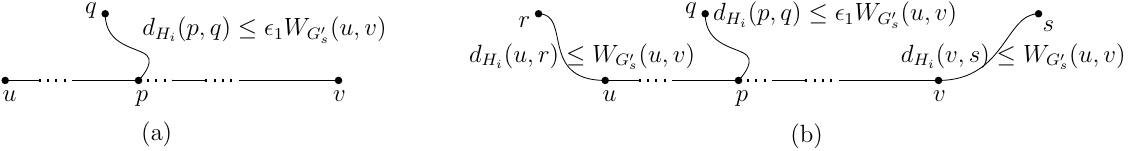}
\caption{Illustration for the set up for (a)  Lemma~\ref{lem:neighbor_improvement} and (b) Lemma~\ref{lem:neighbor_improvement_6}. The missing edges are shown in dotted lines.} \label{fig:lemsetup}
\end{figure}

\begin{lemma}\label{lem:neighbor_improvement}
    Let $H_i$ be a subgraph of $G'_s$ and let $P$ be the shortest path between $u$ and $v$ in $G'_s$. Let $H_{i+1}$ be the graph obtained from $H_i$ by adding the vertices and edges of $P$ that were missing in $H_i$. 
    Let $p$ be a vertex in $P$ and $q$ be a vertex in $H_i$ such that the distance $d_{H_i}(p,q) \leq \epsilon_1 W_{G'_s}(u,v)$ where $\epsilon_1 > 0$ (see  e.g., Figure~\ref{fig:lemsetup}(a)).
    If $d_{H_i}(u,v) > d_{G'_s}(u,v) + \epsilon W_{G'_s}(u,v)$, where $\epsilon = 2\epsilon_1 + \epsilon_2$, where  $\epsilon_2>0$,  
    then the following hold.
    \begin{enumerate}
        \item $d_{H_{i+1}}(u,q) \leq d_{G'_s}(u,q) + 2\epsilon_1 W_{G'_s}(u,v)$ and $d_{H_{i+1}}(v,q) \leq d_{G'_s}(v,q) + 2\epsilon_1 W_{G'_s}(u,v)$
        \item Either $d_{H_i}(u, q) - d_{H_{i+1}}(u, q) > \frac{\epsilon_2}{2}W_{G'_s}(u,v)$ or $d_{H_i}(v, q) - d_{H_{i+1}}(v, q) > \frac{\epsilon_2}{2}W_{G'_s}(u,v)$.
    \end{enumerate} 
\end{lemma}
\begin{proof}
    Using triangle inequality, we have the following:
    \begin{align*} 
    d_{H_{i+1}}(u,q) & \leq d_{H_{i+1}}(u,p) + d_{H_{i+1}}(p,q) \\
                     & \leq d_{H_{i+1}}(u,p) + \epsilon_1 W_{G'_s}(u,v) \\
                     & = d_{G'_s}(u,p) + \epsilon_1 W_{G'_s}(u,v) \\
                     & \leq d_{G'_s}(u,q) + d_{G'_s}(q,p) + \epsilon_1 W_{G'_s}(u,v) \\
                     & \leq d_{G'_s}(u,q) + d_{H_{i+1}}(q,p) + \epsilon_1 W_{G'_s}(u,v) \\
                     & \leq d_{G'_s}(u,q) + \epsilon_1 W_{G'_s}(u,v) + \epsilon_1 W_{G'_s}(u,v) \\
                     & = d_{G'_s}(u,q) + 2\epsilon_1 W_{G'_s}(u,v)
    \end{align*}
    Similarly, we can show that $d_{H_{i+1}}(v,q) \leq d_{G'_s}(v,q) + 2\epsilon_1 W_{G'_s}(u,v)$. If $d_{H_i}(u, q) - d_{H_{i+1}}(u, q) \leq \frac{\epsilon_2}{2}W_{G'_s}(u,v)$ and $d_{H_i}(v, q) - d_{H_{i+1}}(v, q) \leq \frac{\epsilon_2}{2}W_{G'_s}(u,v)$, then

    \begin{align*}
    d_{H_i}(u,v) & \leq d_{H_i}(u,q) + d_{H_i}(q,v) \\
                 & \leq d_{H_{i+1}}(u,q) + \frac{\epsilon_2}{2}W_{G'_s}(u,v) + d_{H_{i+1}}(q,v) + \frac{\epsilon_2}{2}W_{G'_s}(u,v) \\
                 & = d_{H_{i+1}}(u,q) + d_{H_{i+1}}(q,v) + \epsilon_2 W_{G'_s}(u,v) \\
                 & \leq d_{H_{i+1}}(u,p) + d_{H_{i+1}}(p,q) + d_{H_{i+1}}(q,p) + d_{H_{i+1}}(p,v) + \epsilon_2 W_{G'_s}(u,v) \\
                 & = d_{H_{i+1}}(u,p) + d_{H_{i+1}}(p,v) + d_{H_{i+1}}(p,q) + d_{H_{i+1}}(q,p) + \epsilon_2 W_{G'_s}(u,v) \\
                 & = d_{H_{i+1}}(u,p) + d_{H_{i+1}}(p,v) + 2d_{H_{i+1}}(p,q) + \epsilon_2 W_{G'_s}(u,v) \\
                 & \leq d_{H_{i+1}}(u,p) + d_{H_{i+1}}(p,v) + 2\epsilon_1 W_{G'_s}(u,v) + \epsilon_2 W_{G'_s}(u,v) \\
                 & = d_{H_{i+1}}(u,p) + d_{H_{i+1}}(p,v) + (2\epsilon_1 + \epsilon_2)W_{G'_s}(u,v) \\
                 & = d_{G'_s}(u, p) + d_{G'_s}(p, v) + (2\epsilon_1 + \epsilon_2)W_{G'_s}(u,v) \\
                 & = d_{G'_s}(u, v) + (2\epsilon_1 + \epsilon_2)W_{G'_s}(u,v) 
    \end{align*}

    However, we assumed $d_{H_i}(u,v) > d_{G'_s}(u,v) + (2\epsilon_1 + \epsilon_2)W_{G'_s}(u,v)$. Hence either $d_{H_i}(u, q) - d_{H_{i+1}}(u, q) > \frac{\epsilon_2}{2}W_{G'_s}(u,v)$ or $d_{H_i}(v, q) - d_{H_{i+1}}(v, q) > \frac{\epsilon_2}{2}W_{G'_s}(u,v)$. 

    
    

\end{proof}

By Lemma~\ref{lem:neighbor_improvement}, at each step, 
the shortest path distances from $u$ and $v$ to a set of vertices (e.g., vertex $q$ in the lemma) in the current graph become comparable with their shortest path distances in $G'_s$. Such pairs are referred to as getting \emph{set-off}. 
Furthermore, the shortest path distances in the current graph get smaller 
for some pair of vertices, e.g., $(u,q)$ or $(v,q)$. 
We say either the pair $(u,q)$ or $(v,q)$ gets an \textit{improvement}. We now show that a pair can only get $O(\epsilon_1/\epsilon_2)$ improvements after its set-off.


\begin{lemma}\label{lem:number_of_improvements}
    The total number of improvements of a pair of vertices after set-off is $O(\epsilon_1/\epsilon_2)$.  
\end{lemma}
\begin{proof}
    After the set-off of a particular pair $(u,q)$, 
    the shortest path distance between $u,q$ in the current graph is at most $2\epsilon_1 W_{G'_s}(u,v)$ away from their shortest path distance in $G'_s$. Now each improvement decreases the shortest path distance of $u,q$ in the current graph by at least $\frac{\epsilon_2}{2}W_{G'_s}(w,x)$ for some $w,x \in S$. Since the algorithm iterates in the increasing order of maximum edge weight on the shortest path in $G'_s$, the maximum number of improvements is $O(\epsilon_1/\epsilon_2)$.
\end{proof}

Let $\mathcal{H}$ be the spanner constructed above. We are now ready to bound its \lightness.

\begin{theorem}\label{thm:eps_spnr}
The \lightness of the spanner $\mathcal{H}$ is $O_\epsilon(|S|)$.
\end{theorem}
\begin{proof}
    The weight of $H_0$ is less than or equal to $|S'|\leq|V_H|$. Hence the weight of $H_1$ is $O(|V_H|)$. In the next iterations, the algorithm adds missing edges of the shortest paths if the distance condition is unsatisfied. According to Lemma~\ref{lem:neighborhood_all_weight}, for each missing edge added (on the shortest path $P$) there is a neighborhood of size equal to the weight of the missing edge. 
    Let us assume that the total weight of the missing edges is $Z$. We claim a neighborhood of size $\Omega(Z)$. For the sake of contradiction let us assume that the neighborhood has size $o(Z)$. But then we have a shorter path than the shortest path $P$, a contradiction.
    
    According to Lemma~\ref{lem:neighbor_improvement}, each vertex in this $\Omega(Z)$-size neighborhood gets set-off or improvement. Hence the total edge weight added in the remaining iterations is equal to the total number of set-offs and improvements of vertex pairs from $S \times (V_H \cup V_{H_0})$. Since $|V_{H_0}|\leq 2|S'| = O(|V_H|)$, the number such vertex pairs is $O(|S||V_H|)$. Each of these pairs gets at most one set-off and $O(\epsilon_1/\epsilon_2)$ improvements after set-off. Hence the total count is $O(|S||V_H| + |S||V_H|\epsilon_1/\epsilon_2)$. Assuming $\epsilon_1$ and $\epsilon_2$ are constants, the total  spanner weight is 
    $O_\epsilon(|S||V_H|)$, where  $\epsilon=\epsilon_1/\epsilon_2$. 
\end{proof}

Note that the \lightness bound in Theorem~\ref{thm:eps_spnr} is tight because for a complete graph $G$ with edges of unit weight and $S$ equal to the vertex set, any $+\epsilon W(\cdot,\cdot)$ spanner, where $\epsilon\in (0,1]$, gets $\Omega(n^2)$ edges and thus a \lightness of  $\Omega(|S|)$. 

\subsection{A $+(4+\epsilon) W(\cdot, \cdot)$-spanner}

Some initial steps of our $+(4+\epsilon) W(\cdot, \cdot)$-spanner construction are the same as our $+\epsilon W(\cdot, \cdot)$-spanner construction, i.e., we construct the graphs $H$, $H_s$, $H'$, and $G'_s$. The construction of $H_0$ differs from that in the $+\epsilon W(\cdot, \cdot)$-spanner construction; in the $+\epsilon W(\cdot, \cdot)$-spanner construction, the vertex set of $H_0$ is $S'$ and we add at most one edge incident to each vertex in $S'$. In our $+(4+\epsilon) W(\cdot, \cdot)$-spanner, to construct $H_0$, we sort all the neighbors of $u$ in $V_H$ according to the increasing order of edge weights for each vertex $u \in S$. We add edges to $H_0$ until the total weight of edges added adjacent to $u$ is no more than $d$ (where $d$ is a parameter we will optimize later). We then construct $H_1$ by adding the edges of $H_0$ and $H_s$.

\begin{lemma}\label{lem:neighborhood_all_weight4}
    Let $P$ be a shortest path in $G_s$ between a pair of vertices $u, v \in S$. Let $e=(w,r)$ be an edge of $P$ that does not belong to $H_1$. Then there exists a set of vertices $V_r \subseteq V_{H_1}$ such that $|V_r| = \Omega(\sqrt{d})$ and for each vertex $q \in V_r, d_{H_1}(w,q) \leq weight(e)$.  
\end{lemma}
\begin{proof}
    If the weight of $e$ is no more than $\sqrt{d}$, then since $e$ is a missing edge we added at least $\sqrt{d}$ neighbors of $w$ in $H_0$. Let $V_r$ be those neighbors and then $V_r$ satisfies the desired property.
    Assume now that $weight(e) \geq \sqrt{d}$. By Lemma~\ref{lem:remove_large_weight}, $weight(e) \leq |V_H| \leq weight(H_s)$ 
    and $H'$ is obtained by subdividing the edges of $H_s$ such that the weight of each edge of $H'$ becomes at most one. Since $H_1$ contains all edges of $H'$, we can set $V_r$ to be the set of vertices that has a distance less than or equal to $weight(e)$ in $H_1$. Since the weight of edges of $H'$ is at most 1, $|V_r| = \Omega(\sqrt{d})$.
\end{proof}

We then iteratively consider the shortest paths between pairs of vertices of $S$ and for the pairs that do not satisfy the distance condition, we add edges (that are missing on the shortest path) to the current graph $H_i$. However, since we allow for a larger distance bound for $+(4+\epsilon) W(\cdot, \cdot)$-spanners, our goal is to show that a larger number of pairs of vertices of $G'_s$ gets set off at each iteration. The following lemma would help achieve this goal.

\begin{lemma}\label{lem:neighbor_improvement_6}
  
    Let $H_i$ be a subgraph of $G'_s$ and let $P$ be the shortest path between $u$ and $v$ in $G'_s$.  Let $H_{i+1}$ be the graph obtained from $H_i$ by adding the vertices and edges of $P$ that were missing in $H_i$. Let $p$ be a vertex in $P$ and $q$ be a vertex such that the distance $d_{H_i}(p,q) \leq \epsilon_1 W_{G'_s}(u,v)$ where $\epsilon_1 > 0$.
    Let $r$ and $s$ be two vertices such that $d_{H_i}(u,r) \leq W_{G'_s}(u,v)$ and $d_{H_i}(v,s) \leq W_{G'_s}(u,v)$,  respectively  (see  e.g., Figure~\ref{fig:lemsetup}(b)). 
    If $d_{H_i}(u,v) > d_{G'_s}(u,v) + (4+\epsilon)W_{G'_s}(u,v)$, where $\epsilon = 2\epsilon_1 + \epsilon_2$ and $\epsilon_2>0$,  
      then the following hold.
    \begin{enumerate}
        \item $d_{H_{i+1}}(r,q) \leq d_{G'_s}(r,q) + (2 + 2\epsilon_1) W_{G'_s}(u,v)$ and $d_{H_{i+1}}(s,q) \leq d_{G'_s}(s,q) + (2 + 2\epsilon_1) W_{G'_s}(u,v)$.
        \item Either $d_{H_i}(r, q) - d_{H_{i+1}}(r, q) > \frac{\epsilon_2}{2}W_{G'_s}(u,v)$ or $d_{H_i}(s, q) - d_{H_{i+1}}(s, q) > \frac{\epsilon_2}{2}$ $ W_{G'_s}(u,v)$.
    \end{enumerate}
     
\end{lemma}
\begin{proof} 
    Using triangle inequality, we can show the following:
    \begin{align*} 
    d_{H_{i+1}}(r,q) & \leq d_{H_{i+1}}(r,u) + d_{H_{i+1}}(u,p) + d_{H_{i+1}}(p,q) \\
                     & \leq W_{G'_s}(u,v) + d_{H_{i+1}}(u,p) + \epsilon_1 W_{G'_s}(u,v) \\
                     & = d_{G'_s}(u,p) + (1+\epsilon_1) W_{G'_s}(u,v) \\
                     & \leq d_{G'_s}(u,r) + d_{G'_s}(r,q) + d_{G'_s}(q,p) + (1+\epsilon_1) W_{G'_s}(u,v) \\
                     & \leq W_{G'_s}(u,v) + d_{G'_s}(r,q) + d_{H_{i+1}}(q,p) + (1+\epsilon_1) W_{G'_s}(u,v) \\
                     & \leq W_{G'_s}(u,v) + d_{G'_s}(r,q) + \epsilon_1 W_{G'_s}(u,v) + (1+\epsilon_1) W_{G'_s}(u,v) \\
                     & = d_{G'_s}(r,q) + (2+2\epsilon_1) W_{G'_s}(u,v)
    \end{align*}

    Similarly, we can show that $d_{H_{i+1}}(s,q) \leq d_{G'_s}(s,q) + (2+2\epsilon_1) W_{G'_s}(u,v)$. If $d_{H_i}(r, q) - d_{H_{i+1}}(r, q) \leq \frac{\epsilon_2}{2}W_{G'_s}(u,v)$ and $d_{H_i}(s, q) - d_{H_{i+1}}(s, q) \leq \frac{\epsilon_2}{2}W_{G'_s}(u,v)$, then

    \begin{align*}
    d_{H_i}(u,v) & \leq d_{H_i}(u,r) + d_{H_i}(r,q) + d_{H_i}(q,s) + d_{H_i}(s,v) \\
                 & \leq W_{G'_s}(u,v) + d_{H_{i+1}}(r,q) + \frac{\epsilon_2}{2}W_{G'_s}(u,v) + d_{H_{i+1}}(q,s) + \frac{\epsilon_2}{2}W_{G'_s}(u,v) + W_{G'_s}(u,v) \\
                 & = d_{H_{i+1}}(r,q) + d_{H_{i+1}}(q,s) + (2+\epsilon_2) W_{G'_s}(u,v) \\
                 & \leq d_{H_{i+1}}(r,u) + d_{H_{i+1}}(u,p) + d_{H_{i+1}}(p,q) + d_{H_{i+1}}(q,p) + d_{H_{i+1}}(p,v) \\
                 & + d_{H_{i+1}}(v,s) + (2+\epsilon_2) W_{G'_s}(u,v) \\
                 & = d_{H_{i+1}}(r,u) + d_{H_{i+1}}(u,p) + d_{H_{i+1}}(p,v) + d_{H_{i+1}}(v,s) + d_{H_{i+1}}(p,q) \\
                 & + d_{H_{i+1}}(q,p) + (2+\epsilon_2) W_{G'_s}(u,v) \\
                 & = d_{H_{i+1}}(r,u) + d_{H_{i+1}}(u,p) + d_{H_{i+1}}(p,v) + d_{H_{i+1}}(v,s) + 2d_{H_{i+1}}(p,q) + (2+\epsilon_2) W_{G'_s}(u,v) \\
                 & \leq W_{G'_s}(u,v) + d_{H_{i+1}}(u,p) + d_{H_{i+1}}(p,v) + W_{G'_s}(u,v) + 2\epsilon_1 W_{G'_s}(u,v) + (2+\epsilon_2) W_{G'_s}(u,v) \\
                 & = d_{H_{i+1}}(u,p) + d_{H_{i+1}}(p,v) + (4 + 2\epsilon_1 + \epsilon_2)W_{G'_s}(u,v) \\
                 & = d_{G'_s}(u, p) + d_{G'_s}(p, v) + (4 + 2\epsilon_1 + \epsilon_2)W_{G'_s}(u,v) \\
                 & = d_{G'_s}(u, v) + (4 + \epsilon)W_{G'_s}(u,v) 
    \end{align*}

    However, we assumed $d_{H_i}(u,v) > d_{G'_s}(u,v) + (4 + \epsilon)W_{G'_s}(u,v)$. Hence either $d_{H_i}(r, q) - d_{H_{i+1}}(r, q) > \frac{\epsilon_2}{2}W_{G'_s}(u,v)$ or $d_{H_i}(s, q) - d_{H_{i+1}}(s, q) > \frac{\epsilon_2}{2}W_{G'_s}(u,v)$. 



\end{proof}

After computing $H_1$ and $G'_s$, our algorithm sorts the pairs of vertices in $S$ based on their increasing order of maximum edge weight in the shortest path on $G'_s$; If multiple pairs have the same maximum edge weight, then the algorithm selects the pair with the smaller shortest path length. Then it checks whether the spanner condition is unsatisfied for pairs of vertices according to the sorted order. Let $H_i$ be the current subgraph of $G'_s$. If the condition is unsatisfied, then the algorithm computes $H_{i+1}$ by adding the missing edges at $H_i$ of the shortest path of the unsatisfied pair at $G'_s$. Once all the unsatisfied pairs are considered, we have a valid spanner. 
By Lemma~\ref{lem:neighbor_improvement_6}, 
at each iteration, the shortest path distances from close neighbors of $u$ and $v$ to a set of other vertices become comparable with the shortest path distance of $G'_s$, where $u,v\in S $ are the vertices being considered in that iteration. In other words, for each vertex pairs $(r, q)$ and $(s, q)$ satisfying the conditions of the lemma, $d_{H_{i+1}}(r,q) \leq d_{G'_s}(r,q) + (2+2\epsilon_1) W_{G'_s}(u,v)$ and $d_{H_{i+1}}(s,q) \leq d_{G'_s}(s,q) + (2+2\epsilon_1) W_{G'_s}(u,v)$. We say that the pairs $(r,q)$ and $(s,q)$ get \textit{set-off} when the shortest paths of the pairs become comparable with respect to the shortest path in $G'_s$ for the first time. Also, according to the lemma, the shortest paths get smaller either for $(r,q)$ or for $(s,q)$: $d_{H_i}(r, q) - d_{H_{i+1}}(r, q) > \frac{\epsilon_2}{2}W_{G'_s}(u,v)$ or $d_{H_i}(s, q) - d_{H_{i+1}}(s, q) > \frac{\epsilon_2}{2}W_{G'_s}(u,v)$. We say either the pair $(r,q)$ or $(s,q)$ gets an \textit{improvement}. Similar to Lemma~\ref{lem:number_of_improvements}, we can show the total number of improvements of a pair after set-off is $O(\epsilon_1/\epsilon_2)$.

\begin{theorem}\label{thm:4spnr}
The \lightness of the spanner is $O_\epsilon(|V_H|^{1/3} |S|^{1/3})$.
\end{theorem}
\begin{proof}
    By the construction of $H_1$,  $weight(H_1) = O(|V_H| + |S|d)$.
    In the next iterations, the algorithm adds missing edges of the shortest paths if the distance condition is unsatisfied. According to Lemma~\ref{lem:neighborhood_all_weight}, for each missing edge added, there is a neighborhood of size equal to the weight of the missing edge.
    According to Lemma~\ref{lem:neighbor_improvement_6}, each vertex pair formed from these neighborhoods and the neighborhoods of the endpoints of the shortest path get set-off or improvements. Let the missing edges of the shortest path between $u,v \in S$ be added when we construct $H_{i+1}$ from $H_i$. Let $(u, u')$ and $(v',v)$ be the shortest path edges. Notice that, neither $(u, u')$ nor $(v, v')$ are present in $H_i$, otherwise the algorithm would not select the pair $u,v$ (and we would select the pair $u', v'$ since their shortest path is shorter). According to Lemma~\ref{lem:neighborhood_all_weight4}, there are $\Omega(\sqrt{d})$ vertices that form a pair with the neighborhood of size equal to weight of missing edges $z_i$ in the shortest path such that each of these pairs are either get set-off or improvements. Hence there are at least $\sqrt{d} Z$ improvements where $Z=\sum_i z_i$ is the total weights added in the iterations of the algorithms after constructing $H_1$. On the other hand, the size of the set of vertex pairs that can get set-off or improvements is $O(|V_H|^2)$. Each of these pairs gets at most one set-off and $O(\epsilon_1/\epsilon_2)$ improvements after set-off. If we set up $\epsilon = \epsilon_1 /\epsilon_2$, then $Z=O_\epsilon(|V_H|^2/\sqrt{d})$. Hence the total weight of the spanner is $O_\epsilon(|V_H|^2/\sqrt{d} + |S|d)$. 
    If we set $d = |V_H|^{4/3} / (|S|)^{2/3}$, then the total weight of the spanner is $O_\epsilon(|V_H|^{4/3} |S|^{1/3})$. Hence, the \lightness of this spanner is $O_\epsilon(|V_H|^{1/3} |S|^{1/3})$.
\end{proof}

\subsection{A Sampling-based $+(4+\epsilon) W_{\max}$-spanner}
The initial steps of this algorithm are similar to the previous algorithm.
We first compute the Steiner tree $H$, scaled graph $G_s$, subdivided tree $H'$, initial graph $H_0$, add $H'$ to $H_0$ to achieve $H_1$, subdivide Steiner tree edges in the input graph to achieve $G'_s$. We now consider each pair of vertices $u,v \in S$. Let $P$ be the shortest path between $u, v$ in $G'_s$. Let $x$ be the total weight of missing edges in $H_1$ considering the edges of $P$. If $x<\ell$, we then add all the missing edges where $\ell$ is a parameter we will set later. The total weight added in this step is no more than $|S|^2\ell$.

To handle vertex pairs with $x \geq \ell$, we consider the subpath in $H_1$ at the beginning of $P$ that has missing edges of weight $\ell$. We refer to that subpath as the prefix and similarly, we define the suffix. We then uniformly sample $c \ln n |V_H|/\ell$ vertices from $V_H$ where $c$ is a small constant. We also add the missing edges of prefixes and suffixes. If the prefix and suffix overlap, we then add all the missing edges of the shortest path; otherwise, both prefix and suffix have a neighborhood of size $\ell$. Now the probability that the sampled vertices will not hit the neighborhood of suffix and prefix is equal to $(1-\ell/|V_H|)^{c \ln n |V_H|/\ell} \leq 1/n^c$. Hence with a high probability, the sampled vertices will hit the neighborhood of the suffix and prefix.

We then compute a $+\epsilon W(\cdot, \cdot)$-spanner on the sampled vertices and add the spanner edges in $H_1$. 
We can show that we will achieve a valid $+(4+\epsilon) W_{\max}$-spanner with high probability. 
Consider a vertex pair $u,v \in S$ for which we haven't added all the missing edges. Let $p, q$ be two vertices in the prefix and suffix of $P$ respectively, such that sampled vertices hit their neighborhood. Let $r$ and $s$ be the corresponding neighbor vertices of $p$ and $q$ respectively. Then 
\begin{align*}
d_{H_1}(u,v) & \leq d_{H_1}(u,p) + d_{H_i}(p, r) + d_{H_i}(r, s) + d_{H_i}(s, q) + d_{H_1}(q,v) \\
 & \leq d_{G}(u,p) + W_{\max} + d_{G}(r, s) +\epsilon W_{\max} + W_{\max} + d_{G}(q,v)  \\
 & = d_{G}(u,p) + d_{G}(r, s) + d_{G}(q,v) + (2+\epsilon)W_{\max}  \\
 & \leq d_{G}(u,p) + d_{G}(r, p) + d_{G}(p, q) + d_{G}(q, s) + d_{G}(q,v) + (2+\epsilon)W_{\max} \\ 
 & \leq d_{G}(u,p) + W_{\max} + d_{G}(p, q) + W_{\max} + d_{G}(q,v) + (2+\epsilon)W_{\max}  \\
 & = d_{G}(u,p) + d_{G}(p, q) + d_{G}(q,v) + (4+\epsilon)W_{\max}  \\
 & = d_{G}(u,v) + (4+\epsilon)W_{\max} 
\end{align*}

The weight of $+\epsilon W(\cdot, \cdot)$-spanner on $c \ln n |V_H|/\ell$ sampled vertices can be computed using Theorem~\ref{thm:eps_spnr}. 
Let $V'_H$ be the Steiner tree spanning the shortest path vertices of the sampled vertices. Then the weight of the $+\epsilon W(\cdot, \cdot)$-spanner is $c \ln n |V_H||V'_H|/\ell$. Hence the total edge weight of $H_1$ is $O(|S|^2\ell + c \ln n |V_H||V'_H|/\ell)$. If we set $\ell = \sqrt{c \ln n |V_H||V'_H|}/|S|$, then 
$weight(H_1)$ is $O(|S|\sqrt{c \ln n |V_H||V'_H|})$. 
Hence the \lightness of the spanner is $\tilde{O}_\epsilon(|S|\sqrt{|V'_H|/|V_H|})$. Notice that we can not assume that $|V'_H|$ is arbitrarily small since it depends on $\ell$. We can run a binary search on $(0, |V_H|]$ and select an $\ell$ that is approximately $\sqrt{c \ln n |V_H||V'_H|}/|S|$. If we don't find such an $\ell$ we can compute a $+\epsilon W(\cdot, \cdot)$-spanner on $S$.
\begin{theorem}\label{thm:4spnr_sampled}
The \lightness of the spanner is $\tilde{O}_\epsilon(|S|\sqrt{|V'_H|/|V_H|})$.
\end{theorem}

\section{Multi-Level Spanners}\label{sec:multi-level}

In this section, we study a generalization of the subsetwise spanner problem. We now have a hierarchy of $\ell$ vertex subsets $S_{\ell} \subseteq S_{\ell-1} \subseteq \cdots \subseteq S_1$ where $S_1 \subseteq V$. For each level $i$, we compute a (multiplicative or additive) subsetwise spanner $G_i = (V_i, E_i)$ where the edge sets of these spanners are also nested, e.g., $E_{\ell} \subseteq E_{\ell-1} \subseteq \cdots \subseteq E_1$. This problem is called the multi-level spanner problem. We provide an $e$-approximation algorithm for the multi-level spanner problem which is better than the existing $4$-approximation~\cite{ahmed2023multi}. 

The existing $4$-approximation is achieved by first generating a rounding-up instance where each subset is rounded up to a level that is the nearest power of 2, e.g., $2^0, 2^1, \cdots$. Then a spanner is computed independently for each rounded-up level and merged together to achieve a spanner for the edge level. The merging step provides a feasible solution that is a $2$-approximation of the rounded-up instance. On the other hand, the rounding-up step can increase the cost by at most two times the input instance. Overall, the algorithm is thus a $4$-approximation.

We generalize this algorithm by introducing two parameters, $p$ and $q$. Instead of starting from level $2^0$, we start from level $p^q$; and instead of considering the levels $2^0, 2^1, \cdots$, we consider the levels $p^q, p^{q+1}, \cdots$ for rounding up. Our algorithm is randomized, and we show the expected cost of the output with respect to $q$. We then select the best value of $p$ to achieve the desired approximation ratio of the algorithm.

Our approach has two main steps: the first step rounds up the level of all terminals to a level equal to $p^{q+i}$ where $p>1$, $q \in (0,1]$, and $i$ is a nonnegative integer; the second step computes a solution independently for each rounded-up level and merges all solutions from the highest level to the lowest level. We show that the first step can make the solution at most $\frac{p-1}{\ln p}$ times worse than the optimal solution, and the second step can make the solution at most $\frac{1}{1-1/p}$ times worse than the optimal solution. By optimizing $p$, we achieve an $e$-approximation algorithm. 
We provide the pseudocode of the algorithm below, here $\ell$ is the largest integer such that $p^{q+\ell}\leq k$.

\begin{algorithm}
\caption{\textbf{Algorithm} $k$-level Approximation($G = (V, E)$)}
\begin{algorithmic}
\State // Round up the levels
\For{each terminal $v \in V$}
\State Round up the level of $v$ to the nearest power of the form ${p^{q+i}}$
\EndFor
\State // Independently compute the solutions
\State {Let $\partition_{p^q}, \partition_{p^{q+1}}, \partition_{p^{q+2}}, \cdots, \partition_{p^{q+\ell}}$ be the rounded-up terminal subsets}
\For{each subset $\partition_{p^{q+i}}$}
\State Compute a $1$-level solution on $\partition_{p^{q+i}}$
\EndFor
\State // Merge the independent solutions
\For{$j \in \{p^{q+\ell}, p^{q+\ell-1}, \cdots, p^{q}\}$}
\State Merge the solution of $\partition_j$ to the solutions of lower levels
\EndFor
\end{algorithmic}
\end{algorithm}

The cost of the multi-level solution may increase due to rounding up and the merging steps. The following results provide the costs of these steps.

\begin{lemma}\label{lem:round-up}
    The expected rounding cost is $\frac{p-1}{\ln p} \text{weight}(\OPT)$.
\end{lemma}
\begin{proof}
    Assume that the highest level of a particular edge is equal to $p^{s+i}$ before rounding. Since $s \in (0,1]$, after rounding we will face one of two cases; case 1: $s \leq q$, the highest level of the edge becomes $p^{q+i}$ and case 2: $s>q$, the highest level of the edge becomes $p^{q+i+1}$. Since the lowest level is $p^q$ and $q \in (0,1]$, the expected ratio to rounding cost and optimal cost is:
    $\int_s^1 p^{q-s} dq + \int_0^s p^{q-s+1} dq = \frac{p-1}{\ln p}$.
\end{proof}


We now establish an upper bound for the merging step. We show that this upper bound depends only on $p$, regardless of the value of $q$.

 \begin{lemma}\label{lem:merge}
 If we compute independent solutions of a rounded-up instance and merge them, then the cost of the solution is no more than $\frac{1}{1-1/p} \ \text{weight}(\OPT)$.
 \end{lemma}
 \begin{proof}
 Assume that $k = p^{q+i}$. Let the rounded-up terminal subsets be $\partition_{p^{q+i}}, \partition_{p^{q+i-1}}, $ $ \cdots, \partition_{p^q}$. Suppose we have computed the independent solution and merged them in lower levels. 
 Note that in the worst case the cost of approximation algorithm is $p^{q+i} weight(\ALG_{p^{q+i}}) + p^{q+i-1} weight(\ALG_{p^{q+i-1}}) + \cdots + p^q weight(\ALG_{p^q}) = \sum_{s=0}^i p^{q+s} weight(\ALG_{p^{q+s}})$, and cost of the optimal algorithm is $ weight(\OPT_{p^{q+i}}) + weight(\OPT_{p^{q+i}-1}) + \cdots + weight(\OPT_{1}) = \sum_{s=1}^{p^{q+i}} weight(\OPT_{s})$. We show that $\sum_{s=0}^i p^{q+s} weight(\ALG_{p^{q+s}}) \leq \frac{1}{1-1/p} \sum_{s=1}^{p^{q+i}} weight(\OPT_{s})$. We provide a proof by induction on $i$.
 
 If $i=0$, then we have just one terminal subset $\partition_{p^q}$. The approximation algorithm computes a spanner for $\partition_{p^q}$ and there is nothing to merge. Since the approximation algorithm uses an optimal algorithm to compute independent solutions, weight($\ALG_{p^q}$) $\leq $ weight($\OPT_{p^q}$) $\leq \frac{1}{1-1/p}$ weight($\OPT_{p^q}$) since $p>1$.
 
 We now assume that the claim is true for $i = j-1$ which is the induction hypothesis. Hence $\sum_{s=0}^{j-1} p^{q+s} weight(\ALG_{p^{q+s}}) \leq \frac{1}{1-1/p}\ \sum_{s=1}^{p^{q+j-1}} weight(\OPT_{s})$. We now show that the claim is also true for $i = j$. In other words, we show that $\sum_{s=0}^{j} p^{q+s} weight(\ALG_{p^{q+s}}) \leq \frac{1}{1-1/p}\ \sum_{s=1}^{p^{q+j}}  weight(\OPT_{s})$,  using the facts that an independent optimal solution has a cost lower than or equal to any other solution, and that the input is a rounded up instance. 

 Let $k = p^{q+i}$. Let the rounded-up terminal subsets be $\partition_{p^{q+i}}, \partition_{p^{q+i-1}}, \cdots, \partition_{p^q}$. Suppose we have computed the independent solution and merged them in lower levels. We actually prove a stronger claim, and use that to prove the lemma. Note that in the worst case the cost of approximation algorithm is $p^{q+i} weight(\ALG_{p^{q+i}}) + p^{q+i-1} weight(\ALG_{p^{q+i-1}}) + \cdots + p^q weight(\ALG_{p^q}) = \sum_{s=0}^i p^{q+s} weight(\ALG_{p^{q+s}})$. And the cost of the optimal algorithm is $weight(\OPT_{p^{q+i}}) + weight(\OPT_{p^{q+i}-1}) + \cdots +  weight(\OPT_{1}) = \sum_{s=1}^{p^{q+i}} weight(\OPT_{s})$. We show that $\sum_{s=0}^i p^{q+s} weight(\ALG_{p^{q+s}}) \leq \frac{1}{1-1/p} \sum_{s=1}^{p^{q+i}}  weight(\OPT_{s})$. We provide a proof by induction on $i$.
 
 Base step: If $i=0$, then we have just one terminal subset $\partition_{p^q}$. The approximation algorithm computes a spanner for $\partition_{p^q}$ and there is nothing to merge. Since the approximation algorithm uses an optimal algorithm to compute independent solutions, weight($\ALG_{p^q}$) $\leq $ weight($\OPT_{p^q}$) $\leq \frac{1}{1-1/p}$ weight($\OPT_{p^q}$) since $p>1$.
 
 Inductive step: We assume that the claim is true for $i = j-1$ which is the induction hypothesis. Hence $\sum_{s=0}^{j-1} p^{q+s} weight(\ALG_{p^{q+s}}) \leq \frac{1}{1-1/p}\ \sum_{s=1}^{p^{q+j-1}} weight(\OPT_{s})$. We now show that the claim is also true for $i = j$. In other words, we have to show that $\sum_{s=0}^{j} p^{q+s} weight(\ALG_{p^{q+s}}) \leq \frac{1}{1-1/p}\ \sum_{s=1}^{p^{q+s}} weight(\OPT_{s})$. We know,
 \begin{align*}
\sum_{s=0}^{j} p^{q+s}  & weight(\ALG_{p^{q+s}})\\ &= p^{q+j} weight(\ALG_{p^{q+j}}) + \sum_{s=0}^{j-1} p^{q+s} weight(\ALG_{p^{q+s}}) \\
 &\leq p^{q+j} weight(\OPT_{p^{q+j}}) + \sum_{s=0}^{j-1} p^{q+s} weight(\ALG_{p^{q+s}}) \\
 &= \frac{p^{q+j}}{p^{q+j}-p^{q+j-1}} \times (p^{q+j}-p^{q+j-1}) weight(\OPT_{p^{q+j}}) + \sum_{s=0}^{j-1} p^{q+s} weight(\ALG_{p^{q+s}}) \\
 &= \frac{p^{q+j}}{p^{q+j}-p^{q+j-1}} \sum_{s=p^{q+j-1}}^{p^{q+j}} weight(\OPT_{s}) + \sum_{s=0}^{j-1} p^{q+s} weight(\ALG_{p^{q+s}}) \\
 & \leq \frac{1}{1-1/p} \sum_{s=p^{q+j-1}}^{p^{q+j}} weight(\OPT_{s}) + \frac{1}{1-1/p} \sum_{s=1}^{p^{q+j-1}} weight(\OPT_{s}) \\
 & = \frac{1}{1-1/p}\ \sum_{s=1}^{p^{q+j}}  weight(\OPT_{s}) \\
 \end{align*} 
 
 Here, the second equality is just a simplification. The third inequality uses the fact that an independent optimal solution has a cost lower than or equal to any other solution. The fourth equality is a simplification, the fifth inequality uses the fact that the input is a rounded up instance. The sixth inequality uses the induction hypothesis. 
 \end{proof}

 Lemma~\ref{lem:round-up} and Lemma~\ref{lem:merge} show that Algorithm $k$-level approximation provides a multi-level spanner with cost no more than $\frac{p-1}{\ln p} \frac{1}{1-1/p} weight(OPT) = \frac{p}{\ln p} weight(OPT)$. The minimum value of $\frac{p}{\ln p}$ is $e$ when we put $p=e$.

 \begin{theorem}\label{thm:multi-level}
     Algorithm $k$-level approximation computes a $k$-level spanner with expected weight at most $e \cdot weight(OPT)$.
 \end{theorem}

\section{Conclusion}\label{sec:conclusion}
{\color{red}We provided comparisons between Steiner-lightness and \lightness.}
~We designed subsetwise $+\epsilon W(\cdot, \cdot)$ and $+(4+\epsilon) W(\cdot, \cdot)$ spanners with $O_\epsilon(|S|)$  and $O_\epsilon(|V_H|^{1/3} |S|^{1/3})$ \lightness guarantees, respectively. Both spanners can be constructed deterministically. Using a sampling technique we provided the construction for a subsetwise $+(4+\epsilon) W_{\max}$  spanner. 
We also gave an $e$-approximation algorithms for computing multi-level spanners. It would be interesting if the \lightness bound of our subsetwise spanners can be reduced to sublinear bounds, even for some specific graph classes.

\bibliographystyle{plain}
\bibliography{references}

\newpage

\appendix

\end{document}